\documentclass[%
nofootinbib,
 aps,12pt,
]{revtex4-2}
\usepackage{setspace}
\usepackage{graphicx,amsmath,amssymb}
\usepackage{dcolumn}
\usepackage{bm}
\renewcommand\thesection{\arabic{section}}

\newcommand{\avg}[1]{\left< #1 \right>} 
\usepackage{graphicx,color,xcolor}
\usepackage{epstopdf}
\usepackage[subrefformat=parens,labelformat=parens]{subfig}
\usepackage[shortlabels]{enumitem}
\usepackage{rotating}
\usepackage[breaklinks=true,colorlinks=true,linkcolor=blue,urlcolor=blue,citecolor=blue]{hyperref}
\usepackage{amsmath,amsfonts,amsthm}
\newtheorem{theorem}{Theorem}[section] 
\newtheorem{proposition}[theorem]{Proposition}
\newtheorem{definition}[theorem]{Definition}
\newtheorem{lemma}[theorem]{Lemma}
\newtheorem{corollary}[theorem]{Corollary}
\newtheorem{remark}[theorem]{Remark}
\usepackage[utf8]{inputenc}
\usepackage{tikz-feynman}
\tikzfeynmanset{compat=1.0.0}
\newcommand{\vol}{\mathrm{vol}}
\newcommand{\newt}{\mathrm{Newt}}
\newcommand{\conv}{\mathrm{conv}}
\newcommand{\RR}{\mathbb{R}}
\newcommand{\PP}{\mathbb{P}}

\newcommand{\NN}{\mathbb{N}}
\newcommand{\ZZ}{\mathbb{Z}}
\newcommand{\QQ}{\mathbb{Q}}
\usepackage{enumitem}
\newcommand{\dd}{\partial}

\newcommand{\U}{\mathcal{U}}
\newcommand{\F}{\mathcal{F}}
\newcommand{\G}{\mathcal{G}}

\newcommand{\orcid}[1]{\href{https://orcid.org/#1}{\includegraphics[width=10pt]{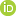}}}

\newtheorem*{theorem*}{Theorem}
\makeatletter
\let\@wraptoccontribs\wraptoccontribs
\def\p@subsection{}
\makeatother
\begin{document}


\title{Cohen-Macaulay Property of Feynman Integrals}

\author{Felix Tellander\orcid{0000-0001-6418-8047}}
\email{felix@tellander.se}
\affiliation{Deutsches Elektronen-Synchrotron DESY, Notkestr.~85, 22607 Hamburg, Germany}

\author{Martin Helmer\orcid{0000-0002-9170-8295}}
 \email{mhelmer@ncsu.edu}
\address{Mathematical Sciences Institute, The Australian National University,\\ Canberra, Australia}
\author{\small (with Appendix by Uli Walther, {\em Department of Mathematics, 
Purdue University})\vspace{8mm}}
 \email{walther@purdue.edu}
\date{\today}

\begin{abstract}\noindent{\bf Abstract:}
The connection between Feynman integrals and GKZ $A$-hypergeometric systems has been a topic of recent interest with advances in mathematical techniques and computational tools opening new possibilities; in this paper we continue to explore this connection. To each such hypergeometric system there is an associated toric ideal, we prove that the latter has the Cohen-Macaulay property for two large families of Feynman integrals. This implies, for example, that both the number of independent solutions and dynamical singularities are independent of space-time dimension and generalized propagator powers. Furthermore, in particular, it means that the process of finding a series representation of these integrals is fully algorithmic.

\end{abstract}

\maketitle
\section{Introduction}\label{sec:Intro} Much of our understanding of physical amplitudes in quantum field theory is tied to their perturbative expansion in terms of Feynman diagrams. This makes Feynman diagrams and their associated integrals central objects in quantum field theory \cite{Folland,Weinberg:1995mt,Peskin:1995ev}. The analytic view of Feynman integrals is as old as the integrals themselves, e.g.~to guarantee causality they are often continued into the complex plane in a predetermined manner. An algebraic viewpoint is not as common in physics, even though it has been known for a some time \cite{Kashiwara1976}, see also \cite{GKZ1}. Recently the algebraic methods of Gelfand, Kapranov and Zelevinski \cite{GKZ1,GKZ1b,GKZ2,gelfand2008discriminants}, using what are now called GKZ $A$-hypergeometric systems, in tandem with the Lee-Pomeransky representation of Feynman integrals \cite{Lee2013} have attracted interest (see e.g.~\cite{delaCruz:2019,Klausen2019,FENG2020,Klemm2019,Bonisch2020} also \cite{Kalmykov:2020cqz}), partially due to the computational utility of this perspective. In this paper we focus on the study of Feynman integrals using this GKZ $A$-hypergeometric system point of view. 

Throughout this paper we will {assume that the underlying Feynman graph $G$ is two-edge connected, or in common physics terminology, $G$ is one particle irreducible (1PI)}. This means that at least two edges in the graph have to be cut for the graph to become disconnected. This is not a substantial restriction from a physical point of view as any connected amplitude can be factored into its 1PI components. Moreover, all integrals are assumed to be dimensionally regularized with generalized dimension $D$.

More precisely we consider scalar Feynman integrals arising from a 1PI Feynman diagram, i.e.~a graph $G=(V,E)$ where each edge is labeled with a mass $m_e$, momentum $k_e$, and propagator $1/(k_e^2-m_e^2)$ and certain vertices are labeled with a momentum $p(v)$. This set of distinguished vertices are called external vertices, $V_{\rm ext}$, and are required to satisfy momentum conservation $\sum_{v\in V_{\rm ext}} p(v)=0$.
Such integrals can be converted to the {\em Lee-Pomeransky} form, which is the standard form we will use here. For a graph $G$ we work over $\RR^{|E|}$ where $|E|$ is the size of the edge set $E$ of the graph $G$. We will also define the {Symanzik} polynomials $\mathcal{U}$ and $\mathcal{F}$ associated to $G$, cf.~\cite{Bogner:2010kv}. The polynomial $\mathcal{U}$ is obtained by summing over all spanning trees in $G$ and for each such tree adding a monomial consisting of all variables whose edge is not in the spanning tree, to obtain $\mathcal{F}$ we sum a polynomial depending on $\mathcal{U}$ with one obtained by summing over spanning two-forests of $G$. Given a spanning two-forest $F=(T,T')$ of $G$ set $p(F)=\sum_{v\in T\cap V_{\rm ext}}p(v)$. In symbols the  {\em Symanzik polynomials}  are:\begin{align}
    \mathcal{U}&=\sum_{\substack{~~~\,T {\rm \;a \; spanning} \\ {\rm tree \; of \; }G}}\;\;\;\prod_{e\not\in T}x_e,\label{eq:U_Symanzik}\\
    \mathcal{F}&=\mathcal{U}\sum_{e\in E}m_e^2x_e+ \sum_{\substack{F {\rm \;a \; spanning} \\ {\rm 2-forest \; of \; }G}}\;\;\; |p(F)|^2\prod_{e\not\in F}x_e,\label{eq:F_Symanzik}
\end{align} where $m_e$ is the mass associated to the edge $e$ and $|p(F)|^2$ is obtained from the Wick rotation of the Lorentz form $p(F)^2\to -|p(F)|^2$. If the Wick rotation is undone, we consider the Euclidean reagion s.t. $p(F)^2<0$ for every $F$.

Our main result is a theorem stating that in many cases the Newton polytope $P=\newt(\U+\F)$ (cf.~\eqref{eq:A_mat_G}) associated to a Feynman integral is normal. This proves a weaker version of the conjecture about existence of unimodular triangulations proposed in \cite{Klausen2019} for our considered classes of diagrams. When working with Feynman integrals from the GKZ $A$-hypergeometric system perspective we will also associate an ideal $I_A$ to such a system. Our main result will directly imply that this ideal $I_A$ is Cohen-Macaulay; this in turn has several important theoretical and computational consequences which are discussed in more depth in Subsection~\ref{sec:FeynmanGKZ}.  
\begin{theorem}[Main Theorem]
Let $G=(V,E)$ be a Feynman diagram with associated {Symanzik} polynomials $\mathcal{U}$ and $\mathcal{F}$. Set $\mathcal{G}=\mathcal{U}+\mathcal{F}$, then the Newton polytope $P_G=\newt(\G)$ is normal if either
\begin{itemize}
    \item $m_e\neq 0$ for all $e\in E$, or\\
    \item $m_e= 0$ for all $e\in E$ and every vertex is connected to an external off-shell leg, i.e. $p_v^2\neq 0$ for every $v\in V=V_\mathrm{ext}$.
\end{itemize}
\label{thm:MainTheoremIntro}   
\end{theorem}
\noindent The second case especially includes all polygon diagrams like the triangle, box or pentagon.

We prove this theorem in two parts, the massive case is treated in Theorem \ref{thm:MainThm} and the massless case in Theorem \ref{thm:MainThmMassless}. In short, this result means that not only can we expect the hypergeometric systems associated to a Feynman diagram to have desirable mathematical properties, but additionally we can expect that the associated Gr\"obner deformation will be straightforward to compute, allowing us to obtain series solutions effectively in an algorithmic manner. 

\subsection{Feynman Integrals and Hypergeometric Systems}\label{sec:FeynmanGKZ}
Let $b\in \ZZ_{\geq 0}^{|E|}$ be an integral vector, and $D\in \RR$; after conversion to {Lee-Pomeransky} form the Feynman integral associated to the graph $G$ is the integral $\mathfrak{I}_G(D, b)$ given by \begin{equation}
\mathfrak{I}_G(D,b):=\frac{\Gamma(D/2)}{\Gamma(D/2-\varsigma)\Gamma(b_1)\cdots \Gamma(b_{|E|})}\int_{\RR^{|E|}_{>0}}\frac{x_1^{b_1-1 }\cdots x_{|E|}^{b_{|E|}-1}}{\mathcal{G}(x)^{D/2}} dx_1\cdots dx_{|E|}
\label{eq:LeePomeranskyForm}
\end{equation}where $\mathcal{G}=\mathcal{U}+\mathcal{F}$, and $\varsigma:=\sum_{i}b_i- L\cdot D/2$ with $L$ the number of independent cycles in the graph $G$.

Suppose that for a given Feynman diagram $G$ the polynomial $\mathcal{G}$ has the form $\mathcal{G}=\sum_{i=1}^r \tilde{c}_i x^{a_i}$. Note that the $\tilde{c}_i$ are explicitly given constants determined by the momenta, masses and graph structure. To consider this as an $A$-hypergoemtric system we will instead take the coefficients as undetermined parameters and consider $\mathcal{G}=\sum_{i=1}^r c_i x^{a_i}$ as a polynomial in the ring $\QQ(D)[c_1,\dots, c_r][x_1,\dots, x_{|E|}]$, this recovers our original polynomial $\mathcal{U}+\mathcal{F}$ in $\QQ(D)[x_1,\dots, x_{|E|}]$  when we set $c_i=\tilde{c}_i$. We abuse notation and use  $\U$, $\mathcal{F}$, and $\mathcal{G}$ to denote both the polynomials in $\QQ(D)[c_1,\dots, c_r][x_1,\dots, x_{|E|}]$ and the resulting polynomial in $\QQ(D)[x_1,\dots, x_{|E|}]$  when we set $c_i=\tilde{c}_i$.  The polynomial $\mathcal{G}$ determines an $(|E|+1)\times r$ integer matrix $A$ obtained by adding a row of ones above the matrix with column vectors the exponents $a_i$ of $\mathcal{G}$: \begin{equation}
A=A_-\times\{1\}:=\begin{pmatrix}
1& 1 & \cdots  & 1 &1\\
a_1 & a_2& \cdots  & a_{r-1} &a_r
\end{pmatrix}\in \NN^{(|E|+1)\times r},\label{eq:A_mat_G}
\end{equation} where $A_-=\begin{pmatrix}
a_1 & a_2& \cdots  & a_{r-1} &a_r
\end{pmatrix}\in \NN^{|E|\times r}$ is the matrix whose columns are the exponent vectors of $\mathcal{G}$. We will refer to the Newton polytope of $\G$, $\newt(\G)=\conv(\{a_1,\ldots,a_r\})$, defined by the convex hull of the vectors as the  {\em Symanzik polytope}. We suppose this polytope is given in half-space representation as
\begin{equation}
    \newt(\G)=\bigcap_{i=1}^N\left\lbrace\sigma\in\RR^{|E|}\ |\ \avg{\mu_i,\sigma}\leq \nu_i\right\rbrace\label{eq:NewtGHalfSpace}
\end{equation} where $\mu_{i}\in \RR^{|E|}$, $\nu\in \RR^N$.

Now return to considering the Feynman integral $\mathfrak{I}_G(D,b;c)$, which we now take as a function of $c$ since we consider $\mathcal{G}$ as a polynomial in $\QQ(D)[c_1,\dots, c_r][x_1,\dots, x_{|E|}]$. The integral $\mathfrak{I}_G(D,b;c)$ is a special case of a so called {\em Euler-Mellin integral}; it is shown in \cite{Berkesch2014} that such integrals admit a meromorphic continuation, giving \begin{equation}
\mathfrak{I}_G(D,b;c)=\Phi(D,b;c)\prod_{i=1}^N\Gamma(\nu_iD/2-\avg{\mu_i,b})
\end{equation}for some function $\Phi$ entire in $D$ and $b$; note $\nu$, $\mu$ are as in \eqref{eq:NewtGHalfSpace}. 
We will also define a vector $\beta$ determined by the vector $b$ and the value $D$ appearing in the Feynman integral in Lee-Pomeransky form \eqref{eq:LeePomeranskyForm}, that is \begin{equation}\beta=\begin{pmatrix}-D/2\\ -b_1\\ \vdots\\  -b_{|E|}\end{pmatrix}.\end{equation}

The function $\Phi(D,b;c)$ is a GKZ A-hypergeometric function of $c$ and satisfies the GKZ $A$-hypergeometric system $H_A(\beta)$, which we now define. Let $W=\QQ(D)[c_1, \dots, c_r, \dd_1, \dots, \dd_r]$ be a Weyl algebra with  $\dd_i$ denoting the differential operator association to $c_i$ (i.e.~$\dd_i$ acts as differentiation by $c_i$ on a polynomial in $\QQ[c_1, \dots, c_r]$) and let $I_A=\langle\dd^u-\dd^v\;|\; Au=Av\rangle$ be the {\em toric ideal} in $\QQ[ \dd_1, \dots, \dd_r]$ defined by the matrix $A$ as in \eqref{eq:A_mat_G} above; the toric ideal is a prime binomial ideal and such ideals define toric varieties, see \cite[Chapter~4]{sturmfels1996grobner}. Writing $A=[a_{i,j}]$, the system $H_A(\beta)$ is a left ideal $H_A(\beta):=I_A+Z_A(\beta)$ in $W$ where \begin{equation}
Z_A(\beta)=\left\langle 
\sum_{j=1}^r a_{i,j}c_j\dd_j-\beta_i \;|\;i=1,\dots, |E|+1
\right\rangle.    
\end{equation}
Finding a basis consisting of holomorphic functions for the space of solutions to the $A$-hypergeometric system $H_A(\beta)$ gives an expression for $\Phi(D,b;c)$, and hence an expression for the Feynman integral $\mathfrak{I}_G(D,b;c)$. By the Cauchy-Kowalevskii-Kashiwara Theorem (see also \cite[Theorem 1.4.19]{SST}) the dimension of the complex vector space of solutions to the system $H_A(\beta)$ in a neighbourhood of a smooth point is equal to ${\rm rank}(H_A(\beta))$, the {\em holonomic rank} of the ideal $H_A(\beta)$. Results of \cite{adolphson1994hypergeometric,GKZ1}, see also \cite[Theorem~4.3.8]{SST}, tell us that if the toric ideal $I_A$ is Cohen-Macaulay for a given $A$ then ${\rm rank}(H_A(\beta))=(|E|!)\cdot \vol(\conv(A))$ and the singular points where solutions to the system $H_A(\beta)$ do not exist are independent of $\beta$. 

A basis for the solution space to the system $H_A(\beta)$ may be computed using techniques described in \cite[Chapter~3]{SST}. An important step in this computation is finding the {\em Gr\"obner deformation} of $H_A(\beta)$ with respect to a {\em weight vector} $\omega\in \RR^{r}$, denoted ${\rm in}_{(-\omega,\omega)}(H_A(\beta))$. This is also greatly simplified in the case $I_A$ is Cohen-Macaulay since in this case \begin{equation}{\rm in}_{(-\omega,\omega)}(H_A(\beta))=Z_A(\beta)+{\rm in}_\omega(I_A),
\end{equation} \cite[Theorem~4.3.8]{SST}, where the later expression ${\rm in}_\omega(I_A)$ is the {\em initial ideal} (or lead term ideal) of $I_A$. The initial ideal of $I_A$ can be computed directly from a Gr\"obner basis of $I_A$, which is in turn straightforward to obtain using standard methods. We obtain the appropriate weight vectors $\omega$ by computing the Gr\"obner fan of $I_A$ and choosing a (generic) representative vector $\omega$ from each cone in the Gr\"obner fan of $I_A$, an efficient procedure (and accompanying software implementation) for computing this Gr\"obner fan of such a toric ideal is detailed in \cite{huber2000computing}. Gr\"obner fans can also be computed using the package Gfan \cite{gfan}, we make use of this implementation via it's Macualay2 \cite{M2} interface  in Section \ref{sec:example} below.
\subsection{Example}\label{sec:example}
We illustrate this process on the Feynman diagram $G$ shown in Figure \ref{fig:massiveBubble}. For further reading on the techniques employed in our example we recommend the book \cite{SST}.

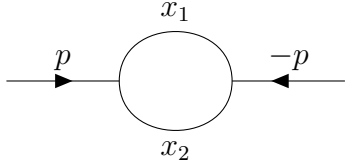
\begin{figure}[h!]
    \centering
    \begin{tikzpicture}[baseline=-\the\dimexpr\fontdimen22\textfont2\relax]
    \begin{feynman}
    \vertex (a);
    \vertex [right = of a] (b);
    \vertex [right = of b] (c);
    \vertex [right = of c] (d);
    \diagram* {
	    (a) --[fermion, edge label=\(p\)] (b) -- [half left,edge label=\({x_1}\)](c) -- [anti fermion, edge label=\(-p\)](d), (c) -- [half left,edge label=\({x_2}\)](b),
};
\end{feynman}
\end{tikzpicture}
    \caption{Feynman diagram $G$ for the massive bubble. There are two edges, with mass $m_1$ associated to the edge $x_1$, and $m_2$ associated to the edge $x_2$ and two vertices with (external) momenta $p$ and $-p$, respectively. }
    \label{fig:massiveBubble}
\end{figure}

In $D$ dimensions the classical presentation for the Feynman integral for the diagram in Figure \ref{fig:massiveBubble} is
\begin{equation}
    \mathfrak{I}=\int \frac{d^Dk}{\pi^{D/2}}\frac{1}{k^2-m_1^2}\frac{1}{(k-p)^2-m_2^2}.
\end{equation}
After Wick-rotating, introducing Feynman parameters and integrating over the loop momenta, this integral can be written in the Lee-Pomeransky form (up to some factors of $\pi$ and $i$), as in \eqref{eq:LeePomeranskyForm} with $b=(1,1)$ as
\begin{align}
\mathfrak{I}_G(D,b)=\mathfrak{I}_G(D,(1,1))&=\frac{\Gamma(D/2)}{\Gamma(D-2)}\int_{\RR_+^2}\ (\U(x)+\F(x))^{-D/2}dx_1dx_2,\\
    \U(x)&=x_1+x_2,\ \F(x)=(m_1^2+m_2^2+|p|^2)x_1x_2+m_1^2x_1^2+m_2^2x_2^2
\end{align}
where $|p|^2>0$ is the Euclidean norm obtained by Wick rotation: $p^2\to-|p|^2$. This integral is a special case of the Euler-Mellin integral which admits the meromorphic continuation
\begin{equation}
    \int_{\RR_+^2}\ (\U+\F)^{-D/2}dx_1dx_2=\Gamma(2-D/2)\Gamma(D-2)\Phi(D)
\end{equation}
where $\Phi(D)$ is an entire analytic function. Treating all the coefficients of the polynomial $\U+\F$ as arbitrary coefficients $c_i$, gives \begin{equation*}
\mathcal{G}(x,c)=\U(x,c)+\F(x,c)=c_1x_1+c_2x_2+c_3x_1x_2+c_4x_1^2+c_5x_2^2.
\end{equation*} Then the function $\Phi(D;c)$ associated to the resulting integral \begin{equation}\label{eq: ex meromorphic}\mathfrak{I}_G(D,1;c)=\int_{\RR_+^2}\ \mathcal{G}(x,c)^{-D/2}dx_1dx_2=\Gamma(2-D/2)\Gamma(D-2)\Phi(D;c)\end{equation}  is $A$-hypergeometric as a function of $c$ and satisfies the $A$-hypergeometric system $H_A(\beta)$ with
\begin{equation}
    A=\begin{pmatrix}
    1&1&1&1&1\\
    1&0&1&2&0\\
    0&1&1&0&2
    \end{pmatrix}=\{1\}\times\newt(\G),\ \ \ \beta=\begin{pmatrix} -D/2 \\ -1\\ -1  \end{pmatrix}.
\end{equation}
Now let $W$ be the Weyl algebra
\begin{equation}
    W=\mathbb{Q}(D)[c_1,\ldots,c_5,\dd_1,\ldots,\dd_5].
\end{equation}
Then the $A$-hypergemoetric system $H_A(\beta)=Z_A(\beta)+I_A$ is the left-ideal in $W$ defined by
\begin{align}
    I_A&=\left\langle \dd_3^2-\dd_4\dd_5,\ \dd_2\dd_3-\dd_1\dd_5,\ \dd_1\dd_3-\dd_2\dd_4 \right\rangle\\
    Z_A(\beta)&=
    \begin{cases}
    c_1\dd_1+c_2\dd_2+c_3\dd_3+c_4\dd_4+c_5\dd_5=\beta_1\\
    c_1\dd_1+c_3\dd_3+2c_4\dd_4=\beta_2\\
    c_2\dd_2+c_3\dd_3+2c_5\dd_5=\beta_3\\
    \end{cases}
\end{align}
where $I_A$ is the toric ideal in $\dd_i$ defined by $A$. Since $m_1$ and $m_2$ are assumed to be non-zero, Theorem \ref{thm:MainTheoremIntro} guarantees that the polytope $\conv(A)$ is normal which in particular implies that $I_A$ is Cohen-Macaulay. For $(-\omega,\omega)\in\RR^{10}$ the Cohen-Macaulay property of $I_A$ guarantees that the Gr\"obner deformation of $H_A(\beta)$ can be decomposed as
\begin{equation}
    \mathrm{in}_{(-\omega,\omega)}(H_A(\beta))=Z_A(\beta)+\mathrm{in}_\omega(I_A).
\end{equation}
 The procedure for constructing a series solutions to $H_A(\beta)$ consists of solving the system given by the Gr\"obner deformation $\mathrm{in}_{(-\omega,\omega)}(H_A(\beta))$ and lifting these solutions to $H_A(\beta)$ by attaching them to a $\Gamma$-series.

The solutions to $\mathrm{in}_{(-\omega,\omega)}(H_A(\beta))$ will be monomials $c^u=c_1^{u_1}\cdots c_5^{u_5}$, $u\in\mathbb{C}^5$. The toric ideal $I_A$ has a Gröbner fan consisting of seven top-dimensional cones, meaning that there are seven distinct initial ideals $\mathrm{in}_\omega(I_A)$. If we choose weight vector $\omega=(0,0,-2,1,1)$, then $I_A$ has the reduced Gröbner basis
\begin{equation}
    \{{(\dd_2\dd_4)} -\dd_1\dd_3,\ (\dd_1\dd_5) -\dd_2\dd_3,\ (\dd_4\dd_5) -\dd_3^2\}
\end{equation}
where the monomials marked with parentheses generates $\mathrm{in}_\omega(I_A)$. If $c^u$ is a solution of the initial system, then the exponent vectors must satisfy
\begin{align}
    u_2u_4=u_1u_5=u_4u_5=0,\\
    \begin{pmatrix}
    1&1&1&1&1\\
    1&0&1&2&0\\
    0&1&1&0&2
    \end{pmatrix}\begin{pmatrix}
    u_1\\u_2\\u_3\\u_4\\u_5
    \end{pmatrix}=\begin{pmatrix}
    -D/2\\-1\\-1
    \end{pmatrix}.
\end{align}
The Cohen-Macaulay property of $I_A$ guarantees that the number of solutions to these equations is the normalized volume of the polytope $\conv(A)$, i.e., these six equations have three solutions:
\begin{align}
    u^{(1)}&=\begin{pmatrix}2-D, & 0, & -1, & D/2-1, & 0    \end{pmatrix}\nonumber\\
    u^{(2)}&=\begin{pmatrix}0, & 2-D, & -1, & 0, & D/2-1    \end{pmatrix}\\
    u^{(3)}&=\begin{pmatrix}1-D/2,& 1-D/2, & D/2-2, & 0, &0 \end{pmatrix}.\nonumber
\end{align}
The three monomials $c^{u^{(1)}},c^{u^{(2)}},c^{u^{(3)}}$ generate the solution space of $\mathrm{in}_{(-\omega,\omega)}(H_A(\beta))$ and can be lifted to solutions of $H_A(\beta)$ as\small
\begin{align}
    \phi^{(i)}=&\sum_{v\in N^{(i)}}\dfrac{\Gamma(u^{(i)}+1)}{\Gamma(u^{(i)}+v+1)}c^{u^{(i)}+v}, \;\;\;{\rm with,}\\
    &N^{(1)}=\left\{v=m(-1,1,1,0,-1)+n(2,-2,0,-1,1),\ m,n\in \ZZ\ |\ m\ge 2n,\ m\le 0,\ n\ge m\right\},\nonumber\\
    &N^{(2)}=\left\{v=m(-1,1,1,0,-1)+n(2,-2,0,-1,1),\ m,n\in \ZZ\ |\ 2n\ge m,\ m\le 0,\ n\le 0\right\},\; {\rm and} \nonumber\\
    &N^{(3)}=\left\{v=m(-1,1,1,0,-1)+n(2,-2,0,-1,1),\ m,n\in \ZZ\ |\ n\le 0,\ n\ge m\right\},\nonumber
\end{align}\normalsize
 where $(-1,1,0,-1)$ and $(2,-2,0,-1,1)$ span the integral kernel of $A$ and the inequalities guarantee that the quotients of $\Gamma$-functions are always well-defined. A solution $\Phi(D;c)$ to the hypergeometric system $H_A(\beta)$ can now be written as $\Phi(D;c)=K_1\phi^{(1)}+K_2\phi^{(2)}+K_3\phi^{(3)}$. The coefficients $K_i$ must be such that the meromorphic continuation on the right hand side of (\ref{eq: ex meromorphic}) matches the left hand side on the domain of convergence of the integral. For example, $K_1$ can be determined by taking the limit $c_2,c_5\to 0$ in (\ref{eq: ex meromorphic}) where $c_2$ and $c_5$ are picked because their respective exponents in $u^{(1)}$ are zero. The integral becomes
\begin{equation}
    \int_{\RR^2_+}\frac{dx_1dx_2}{(c_1x_1+c_3x_1x_2+c_4x_1^2)^{D/2}}=\frac{\Gamma(D-2)\Gamma(1-D/2)}{\Gamma(D/2)}c_1^{2-D}c_3^{-1}c_4^{D/2-1},
\end{equation}
note  the limit is not well-defined for $\Phi(D;c)$ because $c_2$ and $c_5$ appear as denominators, or more precisely, they will have exponents with negative real part\footnote{Note that the form of $N^{(1)}$ guarantees that the limit is well-defined for $\phi^{(1)}$.}. However, the limit is well-defined in the Weyl algebra as the restriction ideal:
\begin{equation}
    \left. H_A(\beta)\right|_{c_2,c_5=0}:=(H_A(\beta)+c_2W+c_5W)\cap \QQ(D)[c_1,c_3,c_4,\dd_1,\dd_3,\dd_4].
\end{equation}
The solution space to this ideal is one-dimensional and spanned by $c_1^{2-D}c_3^{-1}c_4^{D/2-1}$, we thus interpret the limit as $\Phi(D;c)\to K_1c_1^{2-D}c_3^{-1}c_4^{D/2-1}$. Equating this with the explicitly evaluated integral and substituting into (\ref{eq: ex meromorphic}) yields
\begin{equation}
    K_1=\frac{\Gamma(1-D/2)}{\Gamma(D/2)\Gamma(2-D/2)}.
\end{equation}
Similarly we obtain
\begin{equation}
    K_2=K_1,\ \ \ K_3=\frac{\Gamma(D/2-1)\Gamma(D/2-1)}{\Gamma(D/2)\Gamma(D-2)}.
\end{equation}
We have now obtained an explicit series representation for the Feynman integral in one of the seven Gröbner cones, the same procedure can be used to obtain an explicit representation in the other cones.

The paper is organized as follows; in Section \ref{sec:background} we review several definitions and results which will be needed to prove the main theorem, Theorem \ref{thm:MainTheoremIntro}. The main theorem is proved in Section \ref{sec:MainResults}, this proof is separated into two cases, massive and massless. The massive case is treated in Subsection \ref{subsec:massive} and the massless case is treated in Subsection \ref{subsec:massless}.

\section{Background
}\label{sec:background}
In this section we briefly review several definitions and results from different areas of algebra which will be needed in Section \ref{sec:MainResults}. Readers wishing further details should consult books such as \cite{CLO,gelfand2008discriminants,eisenbud2013commutative,MS} on algebraic geometry and \cite{Oxley} on matroid theory. As was discussed in Section \ref{sec:Intro}, in the context of computing series solutions to Feynman integrals, many things become much simpler when the toric ideal $I_A$ associated to the matrix $A$ in \eqref{eq:A_mat_G} has the {\em Cohen-Macaulay} property. Since the matrix $A$ in \eqref{eq:A_mat_G} is always full rank with a row of ones the resulting toric ideal is {\em homogeneous}; recall an ideal $I$ is called homogeneous if it has a homogeneous generating set (equivalently its Gr\"obner basis consists of homogeneous polynomials), i.e.~$I=\langle g_1, \dots, g_t \rangle$ where all monomials appearing in $g_i$ have the same degree.  Hence we will restrict our attention to the case of homogeneous ideals. 

Let $I$ be a homogeneous ideal in a polynomial ring $R=k[z_1,\dots, z_r]$ over a field $k$ of characteristic zero defining a projective variety $X=V(I)\subset \mathbb{P}^{r-1}$ with $d:=\dim(I)=\dim(X)+1$.  Then $d$ homogeneous polynomials $h_1, \dots h_d$ in $R/I$ are called a {\em homogeneous system of parameters} for $R/I$ if $\dim_k(R/I+\langle h_1, \dots, h_d\rangle)<\infty$. We say that a subsequence $h_1, \dots, h_\nu$ is a {\em $(R/I)$-regular sequence of length $r$} if $R/I$ is a free $k[h_1,\dots, h_\nu]$ module, or equivalently if the Hilbert series of $I$, $H_I(z)$, is equal to the Hilbert series of $I+\langle h_1, \dots, h_\nu \rangle$ divided by the polynomial $\prod_{i=1}^\nu(1-z^{\deg(h_i)})$.
\begin{definition}[Cohen-Macaulay]
A homogeneous ideal $I$ in a polynomial $R=k[z_1,\dots, z_r]$ over a field $k$ with $d=\dim(I)$ is {\em Cohen-Macaulay} if there exists a homogeneous system of parameters $h_1, \dots h_d$ such that $h_1, \dots h_d$ is also a $(R/I)$-regular sequence of length $d$.  
\end{definition}

Our interest is in homogeneous toric ideals. That is for a full rank $(|E|+1)\times r$ integer matrix with first row the all ones vector (e.g.~as in \eqref{eq:A_mat_G}) we wish to consider the ideal $I_A=\langle z^u-z^v\;|\; Au=Av \rangle$ in the polynomial ring $k[z_1,\dots,z_{r}]$; this ideal $I_A$ is always a homogeneous prime ideal generated by a finite set of homogeneous binomials. The toric ideal $I_A$ defines  a projective toric variety $X_A=V(I_A)\subset \PP^{r-1}$. We say the semi-group $\NN A$ is {\em normal} if $$\NN A=\ZZ A \cap \RR_{\geq 0} A.$$ For toric ideals a result of Hochster’s \cite{hochster1972rings}, see also \cite[Corollary 1.7.6]{stanley2007combinatorics}, gives us a characterization of the Cohen-Macaulay property of the ideal $I_A$ in terms of the normality of the semi-group $\NN A$. 
\begin{theorem}[Hochster] If the semi-group $\NN A$ is normal then the toric ideal $I_A$ is
Cohen-Macaulay. \label{thm:Hochster}
\end{theorem}

Normality of a configuration of lattice points $A=A_-\times\{1\}$ can be characterised by a combinatorial property of the polytope $P=\conv{(A_-)}$:
\begin{definition}[Normal Polytope] \label{def:NormalPolytope}
A polytope $P$ is called normal, or said to have  the \emph{integer decomposition property}\footnote{This is sometimes called integrally closed.} (IDP), if for any $k\in\NN$
\begin{equation}
    kP\cap\ZZ^d=(k-1)P\cap\ZZ^d+P\cap\ZZ^d.
\end{equation}
\end{definition}

\begin{proposition}[Remark 0.1 of \cite{cox2014integer}]\label{prop:IDP_Normal_SemiGroup}
    A polytope $P$ is IDP if and only if $\NN(P\times\{1\}\cap\ZZ^{d+1})=\RR_{\ge 0}(P\times\{1\})\cap \ZZ^{d+1}$.
\end{proposition}
\noindent This means especially that if all lattice points in $\conv(A_-)$ are column vectors in $A_-$ (which correspond to exponents of monomials in $\G$), i.e.~the set of column vectors of $A_-$ is $\conv(A_-)\cap~\ZZ^d$, the toric ideal $I_A$ will be Cohen-Macaulay if the polytope $P=\conv(A_-)$ is IDP.

   Hence when considering the question of if a toric ideal $I_A$ is Cohen-Macaulay in Section~\ref{sec:MainResults} we will instead seek to prove the stronger sufficient condition that the polytope $P=~{\rm conv}(A_-)$ is normal. 
We now recall two standard constructions in polyhedral geometry.
\begin{definition}
Let $P,Q\subset \RR^d$ be (lattice) polytopes. The {\em Minkowski sum} $P+Q$ is
$$P+Q:=\{p+q\in\RR^d\ |\ p\in P,\ q\in Q\}.$$
The {\em Cayley sum} $P*Q$ is the convex hull of $(P\times\{0\})\cup (Q\times\{1\})$ in $\RR^{d+1}$.
\end{definition}

In Section \ref{sec:MainResults} the notion of an {edge-unimodular} polytope will play a prominent role. Recall that a matrix $M\in\ZZ^{d\times n}$ is said to be \emph{unimodular} if all $d\times d$ minors are either $0$, $1$, or $-1$. A polytope $P$ is called {\em edge-unimodular} if there a unimodular matrix $M$ such that the edges of $P$ are parallel to the columns of $M$. In Section \ref{sec:MainResults} we employ Corollary \ref{thm: edge-unimodular is IDP} which is a direct consequence of the following result of Howard \cite{HOWARD2007,OWR}, see also Danilov and Koshevoy \cite{DANILOV2004}.

\begin{theorem}[Theorem 4.7~of~\cite{HOWARD2007},  cf.~\cite{OWR}]\label{thm: Minkowski sum edge-unimodular}
    Suppose that $M$ is a unimodular matrix and that $P$ and $Q$ are lattice polytopes with edges parallel to the columns of $M$, that is $P$ and $Q$ are both edge-unimodular with matrix $M$. Then
    \begin{equation}
    P\cap\ZZ^d+Q\cap\ZZ^d=(P+Q)\cap\ZZ^d.
    \end{equation}
\end{theorem}

From this theorem we immediately obtain the following result which tells us that to show the projective normality of a toric variety $X_A$ it is sufficient to show that the associated polytope $P=\conv(A)$ is edge-unimodular.
\begin{corollary}
\label{thm: edge-unimodular is IDP}
    If a polytope $P$ is edge-unimodular, then $P$ is IDP.
\end{corollary}
\begin{proof}
Suppose $P$ is edge-unimodular and let $Q=(k-1)P$. Since $Q$ is just a dilation of $P$, thus $Q$ is also edge-unimodular and the prerequisites of Theorem \ref{thm: Minkowski sum edge-unimodular} are met. Hence,
\begin{equation*}
    P\cap\ZZ^d+(k-1)P\cap\ZZ^d=kP\cap\ZZ^d. \qedhere
\end{equation*}
\end{proof}

To prove Theorem \ref{thm:MainThmMassless}, our main result in the massless case, we will need the following result by Tsuchiya \cite[Theorem 0.4]{tsuchiya2019cayley} (see also \cite{Haase2017}) where a complete description of IDP Cayley sums is given. 
\begin{proposition}[Theorem 0.4 of \cite{tsuchiya2019cayley}]\label{prop:TsuchiyaCayleyIDP}
    The Cayley sum $P*Q$ is IDP if and only if $P$ and $Q$ are IDP and also
    \begin{equation}\label{eq:CayeyIDP}
        (a_1P+a_2Q)\cap\ZZ^d=(a_1P\cap\ZZ^d)+(a_2Q\cap\ZZ^d)
    \end{equation}
    for any positive integers $a_1,a_2$.
\end{proposition}

An important class of polytopes, which appear in Section \ref{sec:MainResults}, are the hypersimplices. 
\begin{definition}[Hypersimplex]\label{def:hypersimplex}
The hypersimplex $\Delta(d,k)\subset\mathbb{R}^d$ is the polytope
\begin{equation}
    \Delta(d,k)=\{(x_1,\ldots,x_d)\ |\ 0\le x_1,\ldots,x_d\le 1;\ x_1+\cdots+x_d=k\}.
\end{equation}
\end{definition}
In Section \ref{sec:MainResults} we will also employ several ideas from matroid theory, our main reference for these notions is the book \cite{Oxley}. Below we give several definitions and a theorem which will be of particular importance. 

Given two matroids $M_1,M_2$ on the same ground set $E$, we say that $M_1$ is a {\em quotient} of $M_2$ if every circuit of $M_2$ can be written as a union of circuits in $M_1$.
A pair of matroids $\{M_1,M_2\}$ on the same ground set $E$ form a \emph{flag matroid} if $M_1$ is a quotient of $M_2$. 
In the proof of our main result we will employ the following standard result which tells us that quotients are flipped by duality.
\begin{proposition}[Proposition 7.3.1 of \cite{Oxley}]\label{thm:DualMatridIsQuo}
    Let $M_1,M_2$ be two matroids on $E$, then $M_1$ is a quotient of $M_2$ if and only if $M_2^*$ is a quotient of $M_1^*$.
\end{proposition}

Given a matroid $M$ we may define the associated {\em matroid polytope} $P_M$ to be the convex hull of the indicator vectors of all bases of $M$. We will also wish to associate a polytope to a flag matroid $\{M_1,M_2\}$. 

\begin{definition}\label{def:flag_matroid_polytope}
Let $\{M_1,M_2\}$ be a flag matroid, then the \emph{flag matroid polytope} is defined as the Minkowski sum of the constituent matroid polytopes: $P_{M_1}+P_{M_2}$.
\end{definition}

\section{Normality of Symanzik Polytopes}\label{sec:MainResults}
In this section we prove the main result, namely we show that the polytope associated to entirely massive or entirely massless Feynmann integrals is always IDP, and hence the desirable properties of the associated $A$-hypergeometric system described in Section \ref{sec:Intro} hold. Throughout this section $G=(V,E)$ will be a 1PI Feynman graph as described in Section \ref{sec:Intro}.

\subsection{Massive Case}\label{subsec:massive}
Let $G$ be a 1PI Feynman graph with all internal edges massive, i.e. $m_e\neq 0$ for all $e\in E$. We separate the $\F$-polynomial \eqref{eq:F_Symanzik} as $\F=\F_m+\F_0$ where $\F_0$ is defined by the two-forests and $\F_m$ is given by $\F_m=\U\cdot \sum m_e^2x_e$ with $\U$ as in \eqref{eq:U_Symanzik}. The non-vanishing masses guarantees that every monomial in $\F_0$ will be obtained in $\F_m$, i.e.~writing ${\rm span}(F)$ for the $k$-vector space span of the monomials in a polynomial $F$ over a field $k$ we have $\mathrm{span}(\F_m)\supseteq\mathrm{span}(\F_0)$. To see this, note that every monomial in $\F_0$ can be written on the the form $ux_j$ where $u$ is a monomial in $\U$ and $x_j$ corresponds to one of the edges in the spanning tree defining $u$. If all masses are non-zero, then every $x_j$ will be in the sum $\sum m_e^2x_e$ and thus every monomial in $\F_0$ will be in $\F_m$.

This means that the Newton polytope $P_F:=\newt(\F)$ of $\F$ satisfies 
\begin{equation}\label{eq:P_FandP_U}
    P_F=\newt(\F_m)=P_U+\Delta_E,
\end{equation}
where $P_U:=\newt(\U)$ and $\Delta_E=\Delta(|E|,1)=\conv(e_1,\ldots,e_{|E|})$ is the $(|E|-1)$-dimensional standard simplex in $\RR^{|E|}$; note that the final equality in \eqref{eq:P_FandP_U} follows from the definition of $\F_m=\U\cdot \sum m_e^2x_e$. Let $\mathcal{G}=\U+\F$ and let $\widetilde{\Delta}_E=\conv(0,e_1,\ldots,e_{|E|})$ be the standard simplex with 0 added as a vertex, then $P_G:=\newt(\G)=\newt(\U+\F)$ can be expressed as the sum
\begin{equation}
    P_G=P_U+\widetilde{\Delta}_E.\label{eq:P_G_Def}
\end{equation}
Our goal is then to prove that the polytope $P_G$ is edge-unimodular. \begin{theorem}[Main Theorem I]\label{thm:MainThm}
  Let $P_G$ be the polytope defined in \eqref{eq:P_G_Def}; then the polytope    $P_G$ is edge-unimodular, and hence is IDP. 
\end{theorem}
\begin{proof}
Note that we can construct a co-graphic matroid from $\U$ by taking the matroid whose bases are the complements of the spanning trees of $G$; $P_U$ is the matroid polytope of this matroid. By a classical result of Gelfand, Goresky, MacPherson and Serganova \cite[Theorem 4.1]{GELFAND1987} the edges of a matroid polytope are parallel to $e_i-e_j$, $i\neq j$, where $e_k$ is the $k^{\rm th}$ standard basis in $\RR^{|E|}$. Hence $P_U$ is an edge-unimodular polytope. 

The edges of $\widetilde{\Delta}_E$ are clearly either parallel to $e_i-e_j$ or $e_i$.

The Minkowski sum $P_G=P_U+\widetilde{\Delta}_E$ contains two types of edges: edges parallel to edges of $P_U$ and edges parallel to edges of $\widetilde{\Delta}_E$. This means that $P_G$ has edges in the totally unimodular matrix matrix $(I|A)$ where $I$ is the $(|E|\times |E|)$-dimensional identity matrix and the columns of $A$ consist of vectors which are the columns of some totally unimodular matrix. Hence $P_G$ is edge-unimodular and, by Corollary \ref{thm: edge-unimodular is IDP}, is IDP.
\end{proof}

\begin{remark}
\label{remark:IDP_Poly_is_normal_A}
Lemma \ref{lemma:appendix} in the Appendix below shows that the lattice points in $P_G$ are the same as the columns of $A_-$, (i.e. the exponent vectors of $\G$). Thus $P_G$ being IDP is equivalent to the semi-group $\NN A=\NN(A_-\times\{1\})$ being normal, see Proposition \ref{prop:IDP_Normal_SemiGroup} and the surrounding discussion, which also implies that the toric ideal $I_A$ is Cohen-Macaulay by Hochster's theorem.
\end{remark}

The Symanzik polynomials $\U$ and $\F$ are not only relevant in the Lee-Pomeransky representation but are also used in other parametric representations of Feynman integrals. As observed in the proof of Theorem \ref{thm:MainThm} $P_U$ is a matroid polytope, here we prove a similar result for $P_F$.

\begin{lemma}
 Let $P_F$ be as in \eqref{eq:P_FandP_U}. Then $P_F$ is a flag matroid polytope.
 \end{lemma}
 \begin{proof}
Let $C(|E|)$ be the cycle graph on $|E|$ vertices, i.e.~the graph with $|E|$ vertices connected in a closed chain with $|E|$ edges.  Let $M_{C(|E|)}$ be the associated graphic matroid, that is the matroid whose independent sets are given by the forests of  $C(|E|)$. Then $\Delta_E$ is the matroid polytope of the co-graphic matroid $M_{C(|E|)}^*$. Note that this is a matroid of rank one and whose independent sets are $\mathcal{I}=\{\emptyset,\{1\},\{2\},\ldots,\{|E|\}\}$, thus we see that $M_{C(|E|)}^*=U_{1,|E|}$ where $U_{k,n}$ is the uniform matroid of rank $k$ on $\{1, \dots, n \}$. Let $M_U^*$ be the matroid with matroid polytope $P_U$, this is a matroid on the same ground set $E$ as $U_{1,|E|}$ but has rank $L$ where $L$ is the number of independent cycles in the underlying Feynman graph. 

It is a little easier if we proceed with the dual matroids $M_U$ (the graphical matroid on the underlying Feynman graph) and $U_{|E|-1,|E|}$. 

Note that $U_{|E|-1,|E|}$ only contains one cycle: $\{1,\ldots,|E|\}$. Now, since we have assumed that the underlying Feynman graph is 1PI
then every element in $E$ will be in some cycle of $M_U$. Thus the union of the cycles in $M_U$ will be the cycle in $U_{|E|-1,|E|}$. This means that $M_U$ is a quotient of $U_{|E|-1,|E|}$.

We will now employ Proposition \ref{thm:DualMatridIsQuo}  which tells us that quotients are flipped by duality;
in particular Proposition \ref{thm:DualMatridIsQuo} implies that $U_{1,|E|}$ is a quotient of $M_U^*$ and thus $\{U_{1,|E|},M_U^*\}$ is a flag matroid. 
Since $P_F=P_U+\Delta_E$, where $P_U$, respectively $\Delta_E$, are the matroid polytopes of $M_U^*$, respectively $U_{1,|E|}$, and $\{U_{1,|E|},M_U^*\}$ is a flag matroid, we conclude that $P_F$ is a flag matroid polytope.
 \end{proof}
 
From \cite[Theorem 3.1]{Brovik1997} we have that the edges of a flag matroid polytope are contained in the set of edges of a totally unimodular matrix. This gives us the following corollary. 
\begin{corollary}\label{cor:P_F_EdgeUniMod}
 Let $P_F$ be as in \eqref{eq:P_FandP_U}. Then the edges of $P_F$ are parallel to the columns of a unimodular matrix. 
\end{corollary}

\subsection{Massless Case}\label{subsec:massless}
If all internal edges of a Feynman graph correspond to massless particles, then the $\mathcal{F}$-polynomial \eqref{eq:F_Symanzik} consists only of the sum over spanning 2-forests, $\F=\F_0$, while the $\mathcal{U}$-polynomial \eqref{eq:U_Symanzik} is independent of the internal masses. In order for $x_e$ to be included in a term of $\mathcal{U}$ or $\mathcal{F}$, the corresponding edge $e\in E$ must have been removed. Since an edge can only be removed once, this means that $x_e$ can show up at most once in each term of $\mathcal{U}$ or $\mathcal{F}$. In particular this means that the vertices of $\mathrm{Newt}(\mathcal{U})$ and $\mathrm{Newt}(\mathcal{F})$ are vectors with elements in $\{0,1\}$.

For a Feynman graph with $|E|$ edges and $L$ independent loops, it follows from their definition that $\mathcal{U}$ and $\mathcal{F}$ are homogeneous of degree $L$ and $L+1$ respectively. This in particular means that their Newton polytopes are contained in hyperplanes:
\begin{align}
    \mathrm{Newt}(\mathcal{U})&\subset\{(y_1,\ldots,y_E)\in\mathbb{R}^{|E|}\ |\ y_1+\cdots +y_E=L\},\\
    \mathrm{Newt}(\mathcal{F})&\subset\{(y_1,\ldots,y_E)\in\mathbb{R}^{|E|}\ |\ y_1+\cdots +y_E=L+1\}.
\end{align}
We noted above that the vertices of the Newton polytopes are vectors built of zeros and ones, this together with the fact the polytopes are contained in hyperplanes yields
\begin{align*}
    \mathrm{Newt}(\mathcal{U})\subseteq \Delta(E,L)\ \mathrm{and}\ 
    \mathrm{Newt}(\mathcal{F})\subseteq \Delta(E,L+1),
\end{align*}
i.e.~the Newton polytopes are subsets of hypersimplices (Definition \ref{def:hypersimplex}). Moreover, the fact that $P_U=\newt(\U)$ and $P_{F_0}=\newt(\F_0)$ are in different parallel hyperplanes (which are isomorphic copies of $\RR^{|E|-1}$) means that $P_G$ is their Cayley sum:
\begin{equation}\label{eq:PGasCayleySum}
    P_G=P_U*P_{F_0}.
\end{equation}

For a Feynman graph $G=(V,E)$ with $m_e=0$ for all edges and with all vertices connected to an off-shell external momenta, i.e.~$ p_v^2\neq0,\ v\in V=V_{\rm ext}$, we have the following analog of Theorem \ref{thm:MainThm}.
\begin{theorem}[Main Theorem II]\label{thm:MainThmMassless}
  Let $G=(V,E)$ be a Feynman graph with $m_e=0$ for all $e\in E$ and $V_{\rm ext}=V$, and let $\U$ and $\F_0$ be as above. Then the polytope $P_G=\newt(\U+\F_0)$ is IDP.
\end{theorem}
In light of \eqref{eq:PGasCayleySum} we will apply Proposition \ref{prop:TsuchiyaCayleyIDP} to prove that the Cayley sum $P_G$ is IDP, hence proving  Theorem \ref{thm:MainThmMassless}. To employ Proposition \ref{prop:TsuchiyaCayleyIDP}  we need to show three things:
\begin{itemize}
    \item[\textbf{(i)}] $P_U$ is edge-unimodular (with respect to the unimodular matrix $M$) and hence IDP. As already discussed, this is clear since $P_U$ is a matroid polytope (see the beginning of the proof of Theorem \ref{thm:MainThm}).
    \item[\textbf{(ii)}] $P_{F_0}$ is edge-unimodular (with respect to same unimodular matrix $M$ as in {\bf (i)}) and hence IDP, this is considered in Lemma \ref{lem:P_F0matroid}.
    \item[\textbf{(iii)}] That equation \eqref{eq:CayeyIDP} holds for the pair $P_U$ and $P_{F_0}$, this is considered in Lemma \ref{lemma: cayleyIDP} (keeping in mind $P_{U}$ and $P_{F_0}$ are both edge-unimodular with the same $M$).
\end{itemize}

We now consider {\bf (ii)} above. For each subgraph ${\bf g}\subset G=(V,E)$ we associate the 0/1 vector in $\RR^{|E|}$ indexed by the edges removed from $G$ to get ${\bf g}$, this association is clearly bijective. Given a 0/1 vector $w$ in $\RR^{|E|}$ we will write ${\bf g}_w$ to denote the corresponding subgraph of $G$ obtained by removing the edges corresponding to entries in $w$ with coordinate zero.  
\begin{lemma}\label{lem:P_F0matroid} Let $F_0$ be the set of all spanning two-forest where we view the elements in $F_0$ as 0/1 vectors in $\RR^{|E|}$, i.e.~$F_0$ is the the set of exponent vectors of monomials appearing in $\F_0$, the part of $\F$ in \eqref{eq:F_Symanzik} consisting only of the sum over spanning 2-forests. Then $F_0$ is a set of bases of a matroid. Further the column matrix of the edges of the polytope $P_{F_0}=\conv(F_0)$ forms a totally unimodular matrix. 
\end{lemma}
\begin{proof} Recall that a finite non-empty set $B\subset\ZZ_{\geq0}^n$ is a \emph{base} of a matroid if the following two properties hold:
\begin{itemize}[leftmargin=0.5in]
     \item[\bf{(B1)}] all $u\in B$ have the same norm,
     \item[\bf{(B2)}] if $u,v\in B$ with $u_i>v_i$, then there exists $j\in\{1, \dots, n \}$ with $u_j<v_j$ such that $u-e_i+e_j\in B$, where $e_\ell$ denotes the $\ell^{th}$ standard basis vector.
 \end{itemize} We now show these two properties hold for the set of exponent vectors of $\F_0$; for a vector $u\in \ZZ_{\geq0}^n$ we will use the norm $|u|=u_1+\cdots+u_n$. \\
\textbf{(B1)} The polynomial $\F_0$ is homogeneous of degree $L+1$, where $L$ is the number of independent cycles in $G$, so every $u\in F_0$ satisfies $|u|=L+1$.\\
\textbf{(B2)} Assume $u$ and $v$ are two different elements in $F_0$ such that $u_i>v_i$ for some $i$. Then the graph ${\bf g}_{u-e_i}$ corresponding to the 0/1 vector $u-e_i$ can be one of two types of graphs: 
(a) a spanning tree or (b) a graph with two components, one a tree and the other containing one and only one cycle.
\begin{enumerate}[(a)]
    \item By assumption $u_j<v_j$ for some $j$, since ${\bf g}_{u-e_i}$ is a spanning tree we know that ${\bf g}_{u-e_i+e_j}$ is a spanning two-forest, i.e. $u-e_i+e_j\in F_0$.
    \item For contradiction, assume that for all $j$ such that $u_j<v_j$ we have $u-e_i+e_j\notin F_0$. This assumption means that for any edge $j$ we cut in the graph ${\bf g}_{u-e_i}$ corresponding to the vector $u-e_i$, the cycle in ${\bf g}_{u-e_i}$ will stay intact. Let's do all these cuts; then the graph ${\bf g}_{u-e_i+\sum e_j}$ will still contain the cycle. The resulting graph contains the edge $i$ and all the cuts from $u$ and $v$, since the edge $i$ is in the graph ${\bf g}_v$ corresponding to $v$, this means that the resulting graph is a subgraph of ${\bf g}_v$. But by assumption ${\bf g}_v$ is a spanning two-forest and thus can not contain any cycles. We have a contradiction. 
\end{enumerate}
Applying \cite[Theorem 4.1]{GELFAND1987} gives us that the column matrix of the edges of $P_{F_0}$ forms a totally unimodular matrix and in particular are parallel to $e_j-e_i$. 
\end{proof}

\begin{lemma}\label{lemma: cayleyIDP}
Let $P$ and $Q$ both be edge-unimodular lattice polytopes with edges parallel to the columns of the same unimodular matrix $M$. Then $P$ and $Q$ satisfy \eqref{eq:CayeyIDP}. 
\end{lemma}
\begin{proof}
This follows directly from Theorem \ref{thm: Minkowski sum edge-unimodular} since edge directions are invariant under scaling. In particular $P$ and $Q$ have the same edge directions as $a_1P$ and $a_2Q$.
\end{proof}

\begin{proof}[Proof of Theorem \ref{thm:MainThmMassless}]
As discussed in {\bf (i)} above $P_U$ is edge-unimodular via  \cite[Theorem 4.1]{GELFAND1987} since it is a matroid polytope. By Lemma \ref{lem:P_F0matroid} $P_{F_0}$ is also edge-unimodular (again via  \cite[Theorem 4.1]{GELFAND1987} since it is a matroid polytope). Further we saw in the proof of Lemma \ref{lem:P_F0matroid} that the edges of $P_{F_0}$ are parallel to $e_j-e_i$, $i\neq j$, and saw in the proof of Theorem \ref{thm:MainThm} that the edges of $P_U$ are also parallel to $e_j-e_i$, $i\neq j$. Hence  $P_U$ and  $P_{F_0}$ are both edge-unimodular lattice polytopes with edges parallel to the columns of the same unimodular matrix. It follows by Lemma \ref{lemma: cayleyIDP} that \eqref{eq:CayeyIDP} is satisfied for $P_U$ and $P_{F_0}$. Thus Proposition \ref{prop:TsuchiyaCayleyIDP} applies and $P_G=P_U*P_{F_0}$ is IDP. 
\end{proof}

\begin{remark}
Since $P_U$ and $P_{F_0}$ are maitroid polytopes they have no interior lattice points and additionally they lay in parallel hyperplanes; hence the Cayley sum $ P_G=P_U*P_{F_0}$ also has no interior lattice points and  $P_G\cap \ZZ^{|E|}$ consists only of the vertices of $P_G$. This means that, if the columns of the matrix $A_-$ are the exponent vectors of the polynomial $\G=\U+\F_0$, then the semi-group $\NN A= \NN (A_-\times \{1\})$ is normal, and the associated toric ideal $I_A$ is Cohen-Macualay. 
\end{remark}

\begin{acknowledgments}
We would like to thank Georgios Papathanasiou and Volker Schomerus for many helpful discussions, F.T. especially acknowledge G. Papathanasiou's supervision resulting in this project. We would also like to thank Uli Walther for many helpful correspondences, including sharing with us his proof of Lemma \ref{lemma:appendix} below, and for allowing us to include it here.  
\end{acknowledgments}
\appendix 

\section{A Lemma on Lattice Points}\label{Appendix}
\smallskip
\begin{center}{\bf Author: Uli Walther} \\ Department of Mathematics, 
Purdue University\\
{\em Email}: \texttt{walther@purdue.edu}\end{center}
\medskip
In this section we will consider $G=(V,E)$ as any Feynman graph, not necessarily a 1PI graph, and let $E_m$ denote the set of all edges $e$ with $m_e\neq 0$. 
In order to rule out complications from trivialities we assume that $G$ has at least one edge that is not a loop.
In other words, we assume that the rank of the associated co-graphic matroid is greater than one.

\begin{lemma}\label{lemma:appendix} Let $G=(V,E)$ be any Feynman graph and let $E_m\subset E$ be the set of all edges with non-zero mass, $m_e\neq 0$. Let $\U$ be as in \eqref{eq:U_Symanzik}, with $P_U=\newt(\U)$ and let ${\Delta}_{E_m}$ be the simplex in  $\RR^{|E_m|}$ given by the convex hull of the set of standard basis vectors $\{e_j\;|\; j\in E_m \}$ with  $\widetilde{\Delta}_{E_m}$ being the convex hull of this simplex along with the vector $0\in \RR^{|E_m|}$. The lattice points contained in the polytope $P=P_U+\widetilde{\Delta}_{E_m}$ are exactly those of the form $v + v^\prime$ where $v$ is a vertex of $P_U$ and $v^\prime$
is a vertex
of $\widetilde{\Delta}_{E_m}$.
\end{lemma}
\begin{proof}
Let $M_U^*$ denote the co-graphic matroid of the graph $G$ and $P_U$ its matroid polytope. The lemma clearly holds if $|E| = 1$, and more generally in the case where $E$ is the union of a basis for $M_U^*$ with a set of loops, since then $M_U^*$  has exactly one basis and so $P_U$ is a point and the sum $P_U+\widetilde{\Delta}_{E_m}$ is
a shifted standard simplex.
Let $w$ be a point of $P_U+\widetilde{\Delta}_{E_m}$. Then $w$ can be written as a real linear combination \begin{equation}
w =\sum c_i p_i    \label{eq:w_combo}
\end{equation}where the real numbers $c_i\geq 0$ with $ |c|=\sum c_i=1$ and where each $p_i$
is a vertex of the polytope $P_U+\widetilde{\Delta}_{E_m}$. Let $r:={\rm rank} (M_U^*)$. Note that,  for the vertex $p_i$ in $\RR^{|E_m|}$ the entry-wise sum $|p_i|$ equals either $r$ or $r+1$. It follows that $|w|\in \{ r, r+1\}$.
Now assume in addition that $w$ a lattice point; we must then have $|w|\in \{ r, r+1\}$. Moreover, in either case, since $r$ and $r + 1$ are consecutive integers, the linear combination $\sum c_i p_i$ can only non-trivially involve such $p_i$ with $|w| = |p_i
|$. 

Let $\mathcal{M}_B$ be the set of basis of a matroid $\mathcal{M}$ on ground set $E$ with $v_B\in \ZZ^{|E|}$ denoting the indicator vector of a base $B\in\mathcal{M}_B $; results of White \cite[Theorems~1 and 2]{white1977basis} tell us that the points $(1, a)$ in $\ZZ\times \ZZ^{|E|}$ inside the positive
cone spanned by all pairs $(1, v_B)$,  are precisely the vectors $(1, v_B)$ for $B\in  \mathcal{M}_B$. In our case this result tells us 
that if $|w| = r$ (in which case each $p_i$ with nonzero $c_i$ must have $|p_i
| = r$ and be the indicator vector of a
basis for $M_U^*$ then $w$ is a vertex of $P_U$ , and so $w = w + 0 \in P_U +\widetilde{\Delta}_{E_m}$ is as stipulated in the lemma. We thus assume from now on that $|w| = r + 1$, so $w  \in P_U +\widetilde{\Delta}_{E_m}$.

We consider first the massive case $E_m = E$. Both $P_E=P_U$ and $\widetilde{\Delta}_{E_m}$ are contained in the unit cube, so any
lattice point $w$ of $ P_U +\widetilde{\Delta}_{E_m}$ has coordinate value $x_e(w)$ in the set $\{0, 1, 2\}$, for any $e \in E$. If $x_e(w) = 0$
then all nontrivial terms in \eqref{eq:w_combo} must also satisfy $x_e(p_i) = 0$. Since the set of exponent vectors in $\U$ with
vanishing $e$-coordinate is made of the indicator vectors of the bases for the submatroid of bases of $M_U^*$ that
avoid $e$ (the cographic matroid to the graph derived from $G$ by contracting $e$), it follows by induction on $|E|$
that in this case $w$ is as stipulated in the lemma.

We can therefore assume that there is no $e\in E$ with $x_e(w) = 0$ and so $|w| \geq |E| \geq r$. On the other hand,
we know that $|w| = r + 1$, and so $|E|\in \{r - 1, r\}$. In the latter case, $M_U^*$  is Boolean where the lemma is
straightforward (a {\em Boolean} matroid is one whose only base is the ground set). So we are reduced to checking the case $|E| = r + 1$ which forces $w = (1, \dots , 1)$. In the
massive case $E_m = E$, choose any basis $B$ for $M_U^*$, necessarily of size r. Its indicator vector is the difference
$w - e_f$ for the edge $\{f\} := E -B$ and thus $w = (w - e_f) +e_f \in P_U +\widetilde{\Delta}_{E_m}$ is a sum of vertices as required.

In the non-massive case, $E_m$ is a proper subset of $E$. The previous arguments above show that we are reduced
to investigating $w = (1, \dots, 1)$, and $|E| \in \{r, r + 1\}$.
The Boolean case being trivial, it suffices to show that if $|E| = r + 1$ then $w = (1, \dots , 1)$ is either not in
$P_U +\widetilde{\Delta}_{E_m}$ at all, or equal to the sum of a basis indicator vector of $M_U^*$ with a suitable $e_f$ with $f \in E_m$. If
the latter fails, none of the bases for $M_U^*$ (all of which are of size $r = |E| - 1$) are the complement in $E$ of
an element of $E_m$. In other words, every element of $E_m$ is contained in each basis. In that case, $M_U^*$ is the
matroid sum of the Boolean matroid on $E_m$ (with unique basis $E_m$) with the co-graphic matroid $M_{U_o}^*$, of the graph $G_o$,
on the
ground set $E -E_m$ where $G_o$ is the graph derived from $G$ by deleting the edges of $E_m$. The matroid basis
polytope of $M_{U}^*$ is that of $M_{U_o}^*$
shifted by $\sum_{f\in E_m}e_f$. In other words, we have reduced the problem to the
massless case $E_m=\emptyset$. Then, however, $|w| = r + 1$ implies that $w$ cannot be in $P_U +\widetilde{\Delta}_{E_m}$ .
\end{proof}

\bibliography{Ref}

\end{document}